\documentclass[10pt,twocolumn]{article}

\usepackage{amsmath,amssymb}
\usepackage{amsthm} 
\usepackage{graphicx}
\usepackage{times}
\usepackage{hyperref}
\usepackage{enumitem}
\usepackage{geometry}
\usepackage{tikz}

\theoremstyle{plain}
\newtheorem{theorem}{Theorem}[section]   
\newtheorem{lemma}[theorem]{Lemma}      

\theoremstyle{definition}

\theoremstyle{remark}

\title{\vspace{-1em}Exact \(n + 1\) Comparison Complexity for the N-Repeated Element Problem}
\author{Andrew Au\\[0.2em]
\small Independent Researcher\\[0.2em]
\small Corresponding author: \texttt{cshung@gmail.com}}
\date{\vspace{-1.5em}}

\begin{document}

\twocolumn[
\maketitle
\vspace{-1em}
]

\begin{abstract}
This paper establishes the exact comparison complexity of finding an element repeated $n$ times in a $2n$-element array containing $n+1$ distinct values, under the equality-comparison model with $O(1)$ extra space. 

We present a simple deterministic algorithm performing exactly $n+1$ comparisons and prove this bound \emph{tight}: any correct algorithm requires at least $n+1$ comparisons in the worst case. 

The lower bound follows from an adversary argument using graph-theoretic structure. Equality queries build an \emph{inequality graph} $I$; its complement $P$ (potential-equalities) must contain either two disjoint $n$-cliques or one $(n+1)$-clique to maintain ambiguity. We show these structures persist through $n$ comparisons via a ``pillar matching'' construction, but cannot survive the $(n+1)$st. The matching upper bound comes from a ``triangle'' construction that forces every component of $I$ to be a clique, so each hosts at most one copy of the repeated element and the single untested element must be the answer.

This result provides a concrete, self-contained demonstration of exact lower-bound techniques, bridging toy problems with nontrivial combinatorial reasoning.
\end{abstract}


\section{Introduction}

This paper introduces lower-bound techniques through a concrete, easily stated example. 
The motivating problem is inspired by LeetCode~961 \cite{LeetCode961}.

Consider an array of length \(2n\), containing exactly \(n+1\) distinct values. 
Among these, one element appears \(n\) times, while every other element appears once. 

The task is algorithmically simple---it serves only as a vehicle to explore \emph{how} optimality can be proved.

Several natural approaches illustrate the trade-offs between time and space:

\begin{itemize}[leftmargin=1.2em]
    \item A hash-table approach detects collisions in expected linear time but requires linear extra space.
    \item Sorting the array and examining its middle segment achieves \(O(1)\) extra space but incurs \(O(n \log n)\) time.
    \item Randomized pairwise comparisons often succeed quickly but may, in the worst case, fail to terminate.
\end{itemize}

In this work, I focus on a simple deterministic algorithm and prove its \emph{exact} optimality.

The algorithm proceeds as follows. 
Set aside four elements of the array and partition the remaining \(2n-4\) into \(n-2\) disjoint pairs. 
Compare each pair for equality; if any pair is equal, return that value immediately. 
Otherwise, among the four set-aside elements, compare a \emph{triangle}: pick three of them and compare all three pairs. 
If any of these comparisons reports equality, return that value; 
otherwise, return the fourth set-aside element---the one never involved in any comparison.

This method performs exactly \(n + 1\) comparisons and uses only \(O(1)\) additional space.

The natural question is whether this bound is optimal. 

Under the \emph{equality-comparison model} with constant auxiliary space, 
I prove that any correct algorithm must perform at least \(n + 1\) comparisons. 
Consequently, the above procedure achieves the exact comparison complexity---tight with no additional gap. 

This result provides a concise yet rigorous demonstration of how lower bounds emerge from combinatorial structure, 
even in seemingly simple computational settings.

\section{Computational Model}

To justify a statement such as ``Any correct algorithm must perform at least \(n + 1\) comparisons,'' we must first define precisely what counts as an algorithm. 
Otherwise, one could imagine a nonsensical ``oracle algorithm'':

\begin{quote}
I hand the array to a deity, and the deity immediately returns the correct answer.
\end{quote}

The purpose of a computational model is to exclude such oracles while retaining all legitimate algorithms. 
In general, a good model should be:

\begin{itemize}[leftmargin=1.2em]
    \item broad enough to capture all reasonable algorithms; and  
    \item restrictive enough to make universal lower-bound proofs possible.
\end{itemize}

Although a Turing machine can faithfully represent any algorithm, it is too general for the kind of exact complexity reasoning pursued here. 
Instead, this paper adopts a more specific model: algorithms operate solely through \emph{equality comparisons}.

The allowed operations are as follows:

\begin{itemize}[leftmargin=1.2em]
    \item The algorithm may select any pair of array elements to compare for equality.  
    \item It learns whether the two values are equal or different.  
    \item It may record this information, either explicitly in a small data structure or implicitly as part of its control logic (as in the optimal algorithm from the introduction).  
    \item It may adapt future comparisons based on previous outcomes and eventually output the repeated element.
\end{itemize}

In the original LeetCode formulation, the array elements are numbers, which in principle allows many additional operations: comparisons by magnitude, bitwise manipulations, or arithmetic tests. 
None of these are permitted in this model. 

This restriction is intentional---it simplifies the logical structure of the lower-bound proof. 
Allowing richer operations would make the model more expressive but also render such proofs substantially more difficult, as will be discussed later.

\section{Adversary Argument}

The previous section defined the space of allowable algorithms. 
We now turn to the main question: how can one assert that \emph{any} correct algorithm must perform at least \(n + 1\) comparisons?

Equivalently, the task is to show that every algorithm making at most \(n\) comparisons fails on at least one valid input instance. 

This form of argument is typically established through an \emph{adversary method} \cite{InkuluAdversary}. 
We imagine a two-player game between the algorithm and an adversary, similar to a chess match. 
At each step, the algorithm chooses two elements to compare, and the adversary must respond immediately with either ``equal'' or ``not equal.'' 
Eventually, the algorithm must commit to an answer---its guess for the repeated element. 
The adversary wins if it can produce a concrete input consistent with all previous answers yet for which the algorithm’s output is incorrect.

If such an adversary can always force a win whenever the algorithm performs at most \(n\) comparisons, then the lower bound is proved.

Thus, our goal is to construct a \emph{universal adversary} that guarantees a winning response under this constraint.

The key idea is that the adversary maintains ambiguity by keeping multiple consistent input configurations possible. 
If, after answering all comparison queries, the adversary can still exhibit at least two distinct valid instances---each consistent with the observed answers but having different majority elements---then the algorithm necessarily fails on one of them. 
This property forms the foundation of the lower-bound proof.

\section{Strategy: Graph Interpretation}

How can the adversary maintain multiple consistent answers? 
Initially, before any comparisons have been made, all \(2n\) array positions are unconstrained: 
any subset of \(n\) elements could correspond to the majority element. 

As the algorithm performs equality queries, the adversary naturally responds ``not equal.'' 
Each such answer imposes a new restriction---the two queried elements cannot both belong to the same majority group. 

To reason about these restrictions systematically, we define two complementary graphs.

Let \(I\) denote the \emph{inequality graph}. 
It has a vertex for each array position and an edge between two vertices if and only if the algorithm has tested them for equality. 
Because the adversary always replies ``not equal,'' each comparison adds an edge to \(I\). 

Let \(P\) denote the \emph{potential-equality graph}, defined as the complement of \(I\). 
An edge in \(P\) indicates that the corresponding pair of elements could still be equal---that is, their equality has not yet been ruled out. 

The majority element occurs exactly \(n\) times, and all occurrences of that element are mutually equal. 
Hence, the set of positions containing the majority element forms an \(n\)-clique within \(P\).

For the adversary to preserve ambiguity, \(P\) must contain more than one possible \(n\)-clique.
However, not all such cliques guarantee a win. 
If every \(n\)-clique in \(P\) shares at least one common vertex, the algorithm can safely choose that vertex’s element as its output, since it is consistent with all configurations.

Therefore, the adversary wins precisely when the \(n\)-cliques of \(P\) do \emph{not} all intersect. 
In that case, at least two disjoint possibilities remain for the location of the majority element, ensuring that any algorithm making the current set of comparisons cannot yet be correct.

There are two simple structural conditions that guarantee non-intersection:

\begin{enumerate}[leftmargin=1.2em]
    \item \textbf{Two disjoint \(n\)-cliques.} \\
    If \(P\) contains two vertex-disjoint \(n\)-cliques, the adversary can pick either as the majority group, contradicting any fixed output of the algorithm.
    \item \textbf{One \((n+1)\)-clique.} \\
    An \((n+1)\)-clique contains \(n + 1\) distinct \(n\)-cliques as its subsets, and no single vertex belongs to all of them. 
    Thus, even in this configuration, there is no uniquely guaranteed majority candidate.
\end{enumerate}

These two cases capture exactly the structural scenarios that enable the adversary’s success, and they form the foundation for the tight \(n + 1\) comparison lower bound established later.

\section{Finding a Large Clique}

To apply the previous reasoning, we must ensure that the required \(n\)- or \((n+1)\)-cliques actually exist in the potential-equality graph \(P\). 
How can such cliques be guaranteed?

The key insight comes from viewing \(P\) in terms of its complement \(I\), the inequality graph. 
Let \(C\) denote a minimum vertex cover \cite{VertexCoverIntro} of \(I\). 
If the cover \(C\) is small---that is, if \(|C| < n\)---then the vertices outside the cover, \(V(I) \setminus C\), induce a complete subgraph in \(P\). 
Indeed, no edge of \(I\) connects any two of these vertices, so all possible edges among them exist in \(P\). 
Hence \(P\) contains a clique of size
\[
|V(I) \setminus C| = 2n - |C| > n,
\]
which suffices to guarantee the adversary’s ambiguity.

To bound the size of \(C\), we rely on a simple combinatorial observation.

\begin{lemma}
Let \(G\) be a graph with \(k\) edges. 
Then its minimum vertex cover satisfies \(|C| \le k\).
\end{lemma}

\begin{proof}
Select one endpoint from each edge of \(G\); this set of at most \(k\) vertices meets every edge, and thus forms a valid vertex cover. 
Since a minimum cover cannot be larger, it follows that \(|C| \le k\).
\end{proof}

Applying this lemma to \(I\), if the algorithm performs fewer than \(n\) equality comparisons, then \(I\) has fewer than \(n\) edges. 
Consequently, its minimum vertex cover has size \(|C| < n\), and as shown above, \(P\) contains a clique of size at least \(n + 1\).
Such a structure ensures the existence of multiple consistent instances, implying that any algorithm performing fewer than \(n\) comparisons cannot be correct.

\section{Finding Disjoint Cliques}

What happens when the algorithm performs exactly \(n\) comparisons? 
The vertex-cover lemma now yields \(|C| \le n\), so we can no longer guarantee an \((n+1)\)-clique in \(P\). 
Instead, we seek the second winning condition: two \emph{disjoint} \(n\)-cliques in \(P\).

We first strengthen the vertex-cover observation with a structural characterization.

\begin{lemma}
Let \(G\) be a graph with \(k\) edges whose minimum vertex cover has size exactly \(k\). 
Then the edges of \(G\) form a \emph{matching} \cite{MatchingIntro}: no two edges share a common vertex.
\end{lemma}

\begin{proof}
Suppose two edges share a vertex \(v\). 
Then the set consisting of \(v\) plus one endpoint from each remaining \(k-2\) edges forms a vertex cover of size at most \(k-1\), contradicting minimality. 
Thus, the edges must be vertex-disjoint.
\end{proof}

This structure admits a vivid geometric interpretation. 
Visualize the \(n\) edges of \(I\) as \(n\) ``pillars'': each pillar connects one vertex on the ``table'' (bottom level) to one vertex ``above'' (top level). 
Since the \(n\) edges form a matching, their \(2n\) endpoints are distinct; as \(I\) has exactly \(2n\) vertices, every vertex is an endpoint, giving exactly \(n\) bottom vertices and \(n\) top vertices with no sharing.

In the complement graph \(P\), the absence of edges \emph{within} each level implies \emph{complete connectivity} within each level. 
Thus:
\begin{itemize}[leftmargin=1.2em]
    \item The \(n\) bottom vertices form an \(n\)-clique in \(P\).
    \item The \(n\) top vertices form another \(n\)-clique in \(P\).
\end{itemize}
Moreover, these two cliques are vertex-disjoint by construction.

The adversary can now choose either clique as the majority set---perfectly consistent with all comparison answers but forcing any fixed algorithm output to fail on at least one possibility.

Together with the previous section, this shows that no algorithm can succeed within \(n\) comparisons: at least \(n + 1\) are necessary.

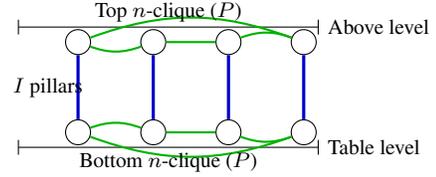
\begin{figure}[t]
\centering
\begin{tikzpicture}[
    vertex/.style={circle, draw, fill=white, minimum size=6pt},
    pillar/.style={very thick, blue!80!black},
    every node/.style={font=\footnotesize}
]

\draw[|-|] (0,-0.2) -- (4,-0.2) node[right] {Table level};
\node[vertex] (b1) at (0.8,0) {};
\node[vertex] (b2) at (1.8,0) {};
\node[vertex] (b3) at (2.8,0) {};
\node[vertex] (b4) at (3.8,0) {};

\draw[|-|] (0,1.4) -- (4,1.4) node[right] {Above level};
\node[vertex] (t1) at (0.8,1.2) {};
\node[vertex] (t2) at (1.8,1.2) {};
\node[vertex] (t3) at (2.8,1.2) {};
\node[vertex] (t4) at (3.8,1.2) {};

\draw[pillar] (b1) -- (t1);
\draw[pillar] (b2) -- (t2);
\draw[pillar] (b3) -- (t3);
\draw[pillar] (b4) -- (t4);

\draw[thick, green!70!black] (b1) to[bend left=20] (b2);
\draw[thick, green!70!black] (b2) -- (b3);
\draw[thick, green!70!black] (b3) to[bend right=20] (b4);
\draw[thick, green!70!black] (b4) to[bend left=20] (b1);

\draw[thick, green!70!black] (t1) to[bend right=20] (t2);
\draw[thick, green!70!black] (t2) -- (t3);
\draw[thick, green!70!black] (t3) to[bend left=20] (t4);
\draw[thick, green!70!black] (t4) to[bend right=20] (t1);

\node at (0.4,0.6) {$I$ pillars};
\node at (2,-0.4) {Bottom $n$-clique ($P$)};
\node at (2,1.6) {Top $n$-clique ($P$)};

\end{tikzpicture}
\caption{Two disjoint $n$-cliques in $P$ from a perfect matching in $I$.
Blue pillars represent ``not equal'' edges; green edges represent cliques in $P$.}
\label{fig:pillars}
\end{figure}

\section{The Upper Bound: An $n+1$ Algorithm}

The lower bound just established shows that \(n\) comparisons never suffice. 
We now show that \(n + 1\) always do, so the bound is tight.

Recall the algorithm from the introduction. 
Set aside four elements and partition the remaining \(2n - 4\) into \(n - 2\) pairs. 
Compare each pair; if any pair is equal, return that value. 
Otherwise, among the four set-aside elements, compare a \emph{triangle}---three elements, all three pairs. 
If any of these is equal, return that value; otherwise return the fourth set-aside element, which was never compared. 
The number of comparisons is \((n - 2) + 3 = n + 1\).

\paragraph{Correctness.} 
If any comparison reports ``equal,'' the two equal positions both hold the repeated element---the only value that occurs more than once---so the returned value is correct. 
Otherwise every answer is ``not equal,'' and we examine the inequality graph \(I\) built by the algorithm. 
Its edges are the \(n - 2\) pair comparisons and the three triangle comparisons, so \(I\) splits into:
\begin{itemize}[leftmargin=1.2em]
    \item \(n - 2\) components, each a compared pair (\(K_2\));
    \item one triangle component (\(K_3\)); and
    \item the single untested vertex.
\end{itemize}
Each of the \(n - 1\) non-trivial components is \emph{complete}: every pair of its positions has been found unequal, so it can host \emph{at most one} copy of the repeated element. 
These \(n - 1\) components therefore hold at most \(n - 1\) copies in total. 
Since the repeated element occurs exactly \(n\) times, at least one copy must fall on the remaining untested position; being a single position, its value \emph{is} the repeated element. 
Returning it is correct.\hfill\(\square\)

\paragraph{Why the adversary cannot survive.} 
The lower-bound adversary keeps ambiguity alive by holding \(I\) to a perfect matching, which produces two disjoint \(n\)-cliques in \(P\) (\autoref{fig:pillars}). 
The triangle defeats precisely this strategy. 
Its three comparisons place three edges on three vertices, but that graph has a minimum vertex cover of size two; 
by the corrected matching lemma---which requires \(k\) edges \emph{and} a minimum vertex cover of size \(k\)---it is not a matching. 
Concentrating three elements into one clique component (at most one repeat) instead of spreading them across the matching is exactly what lets \(n + 1\) comparisons account for all but one copy of the repeated element, collapsing the ambiguity the adversary relied on.

\section{Discussion}

The key property enabling this exact \(n+1\) bound is the restriction to equality comparisons only. 
This model captures only \emph{direct} observations---each comparison adds precisely one edge to the inequality graph \(I\), with no additional inferences possible.

What happens if we allow \emph{magnitude comparisons} (\(<, =\, >\))? 
The transitivity of the total order fundamentally changes the proof structure. 
Suppose the algorithm compares \(a_i < a_j\) and \(a_j < a_k\); it can immediately infer \(a_i < a_k\) \emph{without} an explicit comparison. 
These transitive closures effectively add ``free'' edges to \(I\), dramatically accelerating clique destruction in \(P\).

For example, a single chain of \(n\) magnitude comparisons could establish a total order over \(n+1\) elements, 
forcing all but one to be distinct from the majority---potentially solving the problem in far fewer than \(n+1\) direct comparisons.

This explains why equality-only algorithms require exactly \(n+1\) comparisons, 
while numeric algorithms can do better by leveraging order transitivity. 
The equality model thus serves as an ideal ``toy setting'' for teaching lower-bound techniques: 
simple enough to analyze exhaustively, yet rich enough to demonstrate nontrivial combinatorial structure.

\section*{Acknowledgements}

This research did not receive any specific grant from funding agencies in the public, commercial, or not-for-profit sectors.

\end{document}